\documentclass[12pt]{amsart}
\usepackage{natbib} 
\usepackage{verbatim}
\usepackage{lineno}

\begin{document}

\bibliographystyle{plainnat}
\newcommand{\PS}{\mathcal{P}}
\newcommand{\id}{\mbox{Id}}
\newcommand{\supp}{\mbox{supp}}
\newcommand{\R}{\mathbf{R}}
\newcommand{\E}{\mathbb{E}}
\renewcommand{\P}{\mathbb{P}}
\newcommand{\bR}{\partial\mathbf{R}}
\newcommand{\Rk}{\R^k}  
\renewcommand{\S}{\mathbf{S}}  
\newcommand{\iS}{\S\setminus \S_0}  
\newcommand{\Cp}{\I\Sn }  
\newcommand{\iDelta}{\Delta\setminus\Delta_0}
\def\cL{\mathcal{L}}
\def\cA{\mathcal{A}}

\newtheorem{theorem}{Theorem}
\newtheorem{corollary}{Corollary}
\newtheorem{question}{Question}
\newtheorem{conjecture}{Conjecture}
\newtheorem{lemma}{Lemma}
\newtheorem{proposition}{Proposition}
\newtheorem{definition}{Definition}
\newtheorem{remark}{Remark}
\newtheorem{hypothesis}[theorem]{Hypothesis}
\newtheorem{example}[theorem]{Example}
\newcommand{\nz}{\hfill\break}
\newcommand{\lin}{\mbox{lin }}

\title[Persistence in fluctuating environments]{Persistence in fluctuating environments}

\author[S. J. Schreiber]{Sebastian J. Schreiber}
\address{Department of Evolution and Ecology and the Center for Population Biology\\
University of California, Davis, California 95616}
\email{sschreiber@ucdavis.edu}
\author[M. Bena\"{\i}m]{Michel Bena\"{\i}m}
\address{Institut de Math\'ematiques, Universit\'e de Neuch\^atel, Switzerland}
\email{michel.benaim@unine.ch}
\author[K. Atchad\'{e}]{Kolawol\'{e} A. S. Atchad\'{e}}
\address{Institut de Math\'ematiques, Universit\'e de Neuch\^atel, Switzerland}
\email{kolawole.atchade@unine.ch}
\begin{abstract}
Understanding under what conditions interacting populations, whether they be plants, animals, or viral particles, coexist is a question of theoretical and practical importance in population biology. Both biotic interactions and environmental fluctuations are key factors that can facilitate or disrupt coexistence. To better understand this interplay between these deterministic and stochastic forces, we develop a mathematical theory  extending the nonlinear theory of permanence for deterministic systems to stochastic difference and differential equations. Our condition for coexistence requires that there is a fixed set of weights associated with the interacting populations and this weighted combination of populations' invasion rates is positive for any (ergodic) stationary distribution associated with a subcollection of populations.  Here, an invasion rate corresponds to an average per-capita growth rate along a stationary distribution. When this condition holds and there is sufficient noise in the system, we show that the populations approach a unique positive stationary distribution. Moreover, we show that our coexistence criterion is robust to small perturbations of the model functions.  Using this theory, we illustrate that (i) environmental noise enhances or inhibits coexistence in communities with rock-paper-scissor dynamics depending on correlations between interspecific demographic rates, (ii) stochastic variation in mortality rates has no effect on the coexistence criteria for discrete-time Lotka-Volterra communities, and (iii) random forcing can promote genetic diversity in the presence of exploitative interactions.
\end{abstract}

\maketitle
\centerline{\large Submitted to \emph{Journal of Mathematical Biology}}
\vskip 0.1in

\centerline{\emph{One day is fine, the next is black. --The Clash}}

\section{Introduction}

The interplay between biotic interactions and environmental fluctuations plays a crucial role in determining species richness and genetic diversity~\citep{gillespie-73,chesson-warner-81,turelli-81,chesson-94,ellner-sasaki-96,abrams-etal-98,bjornstad-grenfell-01,kuang-chesson-08,kuang-chesson-09}. For example, competition for limited resources~\citep{gause-34} or sharing common predators~\citep{holt-77} may result in  species or genotypes being displaced. However, random forcing of these systems can reverse these trends and, thereby, enhance diversity~\citep{gillespie-guess-78,chesson-warner-81,abrams-etal-98}. Conversely,  differential predation can mediate coexistence between competitors~\citep{paine-66,holt-etal-94,chesson-kuang-08}, yet environmental fluctuations can disrupt this coexistence mechanism. A fruitful approach to study this interplay  is developing stochastic difference or differential equations and analyzing the long-term behavior of the probability distribution of the population sizes~\citep{turelli-81,chesson-82,gard-84,chesson-ellner-89,ellner-89,gyllenberg-etal-94a,gyllenberg-etal-94b,dcds-07,benaim-etal-08}.

An intuitive approach to the problem of coexistence is given by considering the average per-capita growth rate of a population when rare~\citep{turelli-78,gard-84,chesson-ellner-89}.  When this growth rate is positive, the population can increase and ``invade'' the system. For pairwise interactions, one expects that coexistence is ensured if each population can invade when it is rare and the other population is common. Indeed,
\citet{gard-84} and \citet{chesson-ellner-89} have shown for predator-prey interactions and competitive interactions that ``mutual invasibility'' ensures coexistence in the sense of stochastic boundedness~\citep{chesson-78,chesson-82}:  the long-term distribution of each population is bounded below by a positive random variable. Going beyond pairwise interactions, this mutual invasibility criterion suggests that coexistence should occur if a missing population can invade any subcommunity of the interacting populations. Surprisingly, this criterion  false even for deterministic systems. \cite{may-leonard-75} showed with numerical simulations that  competing species exhibiting a rock-paper-scissor dynamic need not coexist despite every subcommunity being invadable by a missing species.

Starting in  the late 1970s, mathematicians developed a coexistence theory for deterministic models that could handle rock-paper-scissor  type dynamics~\citep{schuster-etal-79,hofbauer-81,hutson-84,butler-waltman-86,hofbauer-so-89,hutson-schmitt-92,jansen-sigmund-98,jde-00,garay-hofbauer-03,nonlinearity-04,jtb-06}. Their notion of coexistence, known as permanence or uniform persistence, ensures that populations coexist despite frequent small perturbations or rare large perturbations~\citep{jansen-sigmund-98,jtb-06}. A sufficient condition for permanence is the existence of a fixed set of weights associated with the interacting populations such that this weighted combination of populations' invasion rates is positive for any invariant measure associated with a sub-collection of populations~\citep{hofbauer-81,jde-00,garay-hofbauer-03}. Conversely, if there is a convex combination of the invasion rates that is negative for all invariant measures associated with a sub-collection of populations, then the community has one or more populations that is extinction prone~\citep{garay-hofbauer-03,nonlinearity-04}.

While environmental stochasticity is often cited as a motivation for the concept of permanence~\citep{hutson-schmitt-92,jansen-sigmund-98}, only recently has the effect of environmental stochasticity on permanent systems been investigated. \citet{benaim-etal-08} found if a deterministic continuous-time model satisfies the aforementioned permanence criterion, then, under a suitable non-degeneracy assumption, the corresponding stochastic differential equation with a small diffusion term  has a positive stationary distribution concentrated on the positive global attractor of the deterministic system. Consequently, permanent systems persist despite continual, but on average small, random perturbations. Conversely, if the deterministic dynamics satisfies the impermanence criterion, then the stochastic dynamics almost surely converges to the boundary of the state space. This asymptotic loss of one or more species occurs even if there is a positive attractor for the underlying deterministic dynamics.

For many systems, stochastic perturbations may not be small and these perturbations may not be best described by stochastic differential equations \citep{turelli-78}. Here, in sections~\ref{sec:discrete} and \ref{sec:conts}, we develop a natural generalization of the permanence criteria for stochastic difference  and differential equations with arbitrary levels of noise. The proofs of these results are presented in the Appendices. In section~\ref{sec:apps}, we develop applications of these results to competitive lottery models, discrete-time Lotka-Volterra models with environmental disturbances, and stochastic replicator equations. In particular, we show how environmental stochasticity can enhance or disrupt diversity in these models.

\section{Discrete time models\label{sec:discrete}}

\subsection{The models}
We study the dynamics of $k$ interacting  populations in a random environment. Let $X^i_t$ denote the density of the $i$-th population at time $t$ and $X_t=(X_t^1,\dots,X_t^k)$ the vector of population
densities at time $t$.\footnote{For sequences of random vectors, we use subscripts to denote time $t$ and superscripts to denote components of the vector. For all other vectors, we use subscripts to denote components of the vector.} To account for environmental fluctuations,   we introduce a random variable $\xi_t$ that represents the state of the environment (e.g. temperature, nutrient availability) at time $t$. The fitness $f_i(X_t,\xi_{t+1})$ of population $i$ at time $t$ depends on the population state and environmental state at time $t+1$. Under these assumptions, we arrive at the following stochastic difference equation:
\begin{equation}\label{eq:discrete}
X_{t+1}=f(X_t,\xi_{t+1})\circ X_t
\end{equation}
where  $f(x,\xi)=(f_1(x,\xi),\dots,f_k(x,\xi))$ is the vector of fitnesses and $\circ$ denotes the Hadamard product i.e. component-wise multiplication.

Regarding \eqref{eq:discrete}, we make four assumptions:
\begin{description}
\item [A1] There exists a compact set $\S$ of $\R^k_+=\{x\in\R^k: x_i\ge 0\}$ such that $X_t\in \S$ for all $t\ge 0$.
\item [A2] $\{\xi_t\}_{t=0}^\infty$ is a sequence of i.i.d random variables independent of $X_0$ taking values in a probability space $E$ equipped with a $\sigma$-field and probability measure $m$.
\item [A3] $f_i(x,\xi)$ are strictly positive functions,  continuous in $x$ and measurable in $\xi.$
\item [A4]  For all $i$, $\sup_{x\in \S}\int (\log f_i(x,\xi))^2 \, m(d\xi)<\infty$
\end{description}
Assumption \textbf{A1} ensures that the populations remain bounded for all time. Assumptions \textbf{A2} and \textbf{A3}  imply that $\{X_t\}_{t=0}^\infty$ is a Markov chain on $\S$ and that $\{X_t\}_{t=0}^\infty$ is {\em Feller}, meaning that $Ph$, as defined below, is continuous whenever $h$  is continuous.  Assumption \textbf{A4} is a technical assumption meet by many models.

\subsection{Some ergodic theory}
In order to state our main results, we introduce some notation and review some basic concepts from ergodic theory. For any Borel set $A\subset \S$ and $x \in \S$, let
\[
\P_x[X_t\in A] = \P[X_t\in A \big|X_0=x].
\]
be the probability $X_t$ is in $A$ given that $X_0=x$. For various notions of convergence, it is useful to consider how the expected value of an ``observable'' (a function $h$ from $\S\to \R$)  depends on the dynamics of $X_t$. Give a bounded or nonnegative measurable function $h : \S \mapsto \R$, define
\[
\E[h(X_1)|X_0=x] = \int h(f(x,\xi)\circ x)m(d\xi)
\]
to be the expected value of $h$ in the next time step given that the current state of the population is $x$.  Let $P$ be the operator on bounded measurable functions defined by $Ph(x)=\E[h(X_1)|X_0=x]$. 

To understand the long-term statistical behavior of the population dynamics, it is useful to introduce  invariant probability measures. Roughly, a probability measure $\mu$ is invariant if the population initially follows the distribution of $\mu$, then it follows this distribution for all time i.e. if $\P[X_0\in A] = \mu(A)$ for all Borel sets $A\subset \S$, then $\P[X_t\in A]=\mu(A)$ for all $t$ and all Borel sets $A\subset \S$. One can phrase this invariance in terms of observables $h:\S\to \R$. If $X_0$ follows the distribution of $\mu$, then the expected value of $h(X_0)$ equals $\int_\S h(x)\,\mu(dx)$ and the expected value of $h(X_1)$ equals $\int_\S Ph(x)\,\mu(dx)$. It follows that a Borel probability measure $\mu$ is  \emph{invariant} for $\{X_t\}_{t=0}^\infty$ or $P$ if
\begin{equation}
\label{eq:inv}
\int_\S h(x) \,\mu(dx) = \int_\S Ph(x) \, \mu(dx)
\end{equation}
for all continuous bounded functions $h:\S \to\R$.  We let $\PS$ denote the space of Borel probability measures on $\S$

If $X_t$ initially follows the distribution of an invariant probability measures $\mu$, then Birkhoff's ergodic theorem implies that the temporal averages of an observable along a population trajectory converges with probability one. More precisely, if $h:\S\to \R$ is a measurable function with $\int_\S | h(x)| \,\mu(dx)<\infty$, then there exists a  measurable function $\widetilde h:\S\to \R$ such that  $\int_\S |\widetilde h(x)|\,\mu(dx)<\infty$ and 
\[
\lim_{t\to\infty} \frac{1}{t}\sum_{s=0}^{t-1} h(X_s) = \widetilde h(X_0)
\]
with probability one. When $\widetilde h$ is a constant function for all bounded measurable $h$, $\mu$ is called an \emph{ergodic probability measure} in which case 
\begin{equation}\label{eq:ergodic}
\lim_{t\to\infty} \frac{1}{t}\sum_{s=0}^{t-1} h(X_s) = \int_\S h(x)\,\mu(dx)
\end{equation}
with probability one. Since $\int_\S h(x)\,\mu(dx)$ corresponds to the expected value of $h(X_0)$, equation \eqref{eq:ergodic} can be interpreted as a law of large numbers for $X_t$. 

While the Birkhoff ergodic theorem provides a relatively complete picture of the long-term statistical behavior of $X_t$, it only does so when $X_0$ is initially distributed like an invariant probability measure. However, as we are interested in the long-term behavior of $X_t$ for any positive initial condition $X_0$, new results are needed that require the concept of an invasion rate.

\subsection{Invasion rates}

The \emph{expected per-capita growth rate at state $x$} of population $i$ is
 \[\lambda_i(x)= \int \log f_i (x,\xi) m(d\xi).\]
 When $\lambda_i(x)>0$, the $i$-th population tends to increase when the current population state is $x$.  When $\lambda_i(x)<0$, the $i$-th population tends to decrease when the current population state is $x$.
 For an invariant probability measure $\mu$, we define the \emph{invasion rate of species $i$ with respect to $\mu$} to be
\[
\lambda_i(\mu)=\int_\S \lambda_i(x)\,\mu(dx)
\]
The following proposition clarifies why $\lambda_i(\mu)$  is called an invasion rate. Its proof is in Appendix A.

\begin{proposition}
\label{prop:invasion}  Let $\mu$ be an invariant  probability measure and $i \in \{1, \ldots, k \}.$ Then there exists a bounded  map $\hat{\lambda}_i:\S\to \R$ such that:
\begin{enumerate}
\item[(i)] With probability one and for $\mu$-almost every $x\in \S$
\[
\lim_{t\to\infty} \frac{1}{t} \sum_{s=0}^{t-1} \log f_i(X_s,\xi_{s+1}) =\hat{ \lambda}_i(x)\mbox{ when }X_0=x;
\]

\item[(ii)] $\int_{\S} \hat{\lambda}_i(x) \mu(dx) = \lambda_i(\mu);$ Furthermore if $\mu$ is ergodic, then $\hat{\lambda}_i(x) = \lambda_i(\mu)$ $\mu$-almost surely.
\item[(iii)] If $\mu(\{x \in \S : \: x_i > 0 \}) = 1$, then $\lambda_i(\mu)=0.$
\end{enumerate}
\end{proposition}

When $\mu$ is ergodic, $\lambda_i(\mu)$ is the long-term time average of the per-capita growth rate of population $i$. Moreover, since each of the set $\{x_i = 0\}$ is invariant under the  dynamics in (\ref{eq:discrete}), there exists a set $\supp(\mu) \subset \{1,\ldots, k\}$ such that  $x_i > 0$ if and only if $i \in \supp(\mu)$ for $\mu$-almost all $x$. One can interpret $\supp(\mu)$ as the set of populations supported by $\mu$. Quite intuitively, Proposition~\ref{prop:invasion} implies that the long-term average of the per-capita growth is zero for all populations supported by $\mu$ i.e.  $\lambda_i(\mu)=0$ for all $i\in \supp(\mu)$.

\subsection{Persistence}
To quantify persistence, there are two ways to think about the asymptotic behavior of $\{X_t\}_{t=0}^\infty$. First, one can ask what is the distribution of $X_t$ far into the future. For example, what is the probability that  the population density of each state is greater than $\epsilon$ in the long term  i.e. $\P[X_t\ge (\epsilon,\dots,\epsilon)]$ for large $t$? The answer to this question provides information what happens across many independent realizations of the population dynamics. Alternatively, one might be interested about the statistics associated with a single realization of the process i.e. a single time series. For instance, one could ask what fraction of the time was the density of each population state greater than $\epsilon$? To answer this question, it is useful to introduce the \emph{occupation measures}
\[\Pi_t = \frac{1}{t} \sum_{s = 1}^t \delta_{X_s}\]
where $\delta_{X_s}$ denotes a Dirac measure at $X_s$ i.e. $\delta_{X_s}(A)=1$ if $X_s\in A$ and $0$ otherwise for any (Borel) set $A\subset \S$. One can interpret $\Pi_t(A)$ as the proportion of time the population spends in $A$ up to time $t$.

Our first theorem addresses persistence from the second perspective. To state this theorem, for $\eta>0$, let $\S_\eta=\{x\in S: x_i\le \eta$ for some $i\}$ be the set of the states where at least one population has an abundance less than or equal to $\eta$. $\S_0$ corresponds to the states where one or more populations is extinct.

\begin{theorem}[Persistence]
\label{thm:persistence}
Assume that one of the following equivalent conditions hold:
\begin{enumerate}
\item[(i)] For all invariant probability measures $\mu$  supported on $ \S_0$,
\[
\lambda_*(\mu) := \max_i  \lambda_i(\mu)>0
\]
\item[(ii)]
There exists $p\in \Delta$ such that
\[
\sum_i p_i \lambda_i (\mu) >0
\]
for all ergodic probability measures $\mu$  supported by $\S_0.$
\end{enumerate}
Then for all $\epsilon>0$ there exists $\eta>0$ such that
\[
\limsup_{t\to\infty} \Pi_t (\S_\eta)\le \epsilon \mbox{ almost surely}
\]
whenever $X_0=x\in \iS$.
\end{theorem}

Theorem~\ref{thm:persistence} implies that fraction of time spent by the populations in $\S_\eta$ goes to zero as $\eta$ goes to zero. Theorem~\ref{thm:persistence}, however, does not ensure that there is a unique positive stationary distribution. For this stronger conclusion, there has to be sufficient noise in the system to ensure after enough time any positive population state can move close to any other positive population state. More precisely, given $\eta > 0,$ we say that $\{X_t\}$ is {\em irreducible} over $\S \setminus \S_\eta$ if there exists a probability measure $\Phi$ on $\S \setminus \S_{\eta}$ such that
for all $x \in \S \setminus \S_\eta$ and every Borel set $A \subset \S \setminus  \S_\eta$ there exists $n \geq 1$  (depending on $x$ and $A$) such that
\[\P_x(X_n \in A) > 0\]
whenever $\Phi(A) > 0.$

\begin{theorem}[Uniqueness]\label{thm:uniqueness} Assume that   $\{X_t\}$ is irreducible over $\S \setminus \S_{\eta}$  for all $\eta > 0$, and that the assumption of Theorem~\ref{thm:persistence} holds. Then there exists a unique invariant probability measure $\pi$ such that $\pi (\S_0)=0$ and
the occupation measures $\Pi_t$ converge almost surely to $\pi$ as $t\to\infty$, whenever $X_0=x\in \iS$.
\end{theorem}

Theorem~\ref{thm:uniqueness} ensures the asymptotic distribution of one realization of the population dynamics is given by the positive stationary distribution $\pi$. Hence, $\pi$ provides information about the long-term frequencies that a population trajectory spends in any part of the population state space. To gain information about the distribution of $X_t$ across many realizations of the population dynamics, we need a stronger irreducibility condition. This stronger condition requires that after a fixed amount of time independent of initial condition, any positive population state can move close to any other positive population state. More precisely, we say that $\{X_t\}$ is {\em strongly irreducible} over $\S \setminus \S_\eta$ if there exists a probability measure $\Phi$ on $\S \setminus \S_\eta,$ an integer $n \geq 1$ and some number $0 < h \leq 1$ such that for all $x \in \S \setminus \S_\eta$ and every Borel set $A \subset  \S \setminus \S_\eta$
\[\P_x(X_n \in A) \geq h \Phi(A).\]
To state the next result given $\mu, \nu  \in \PS$ define
\[\|\mu - \nu\| = \sup_{B} |\mu(B) - \nu(B)|\] where the supremum is taken over all Borel sets $B \subset \iS.$

\begin{theorem}[Convergence in distribution]\label{thm:PHC} In addition to the assumptions of Theorem \ref{thm:uniqueness} assume that  $\{X_t\}$ is strongly irreducible over $\S \setminus \S_{\eta}$ for all $\eta > 0$.  Then
 the distribution of $X_t$ converges to $\pi$ as $t\to\infty$ whenever $X_0=x\in \iS;$ that is
\[ \lim_{t \rightarrow \infty} \|\P_x[X_t \in \cdot] - \pi\| = 0\mbox{ for all }x\in \iS.\]
\end{theorem}

\begin{remark}
{\rm   Suppose that there exists a nonzero continuous function  $\rho : \, \iS  \times \iS \mapsto  \R_+$  an integer $n \geq 1$ and a probability $\nu$ over $\iS$  such that for all $x \in \iS$ and Borel set $A \subset  \iS$
\[\P_x[X_n \in A] \geq \int_A \rho(x,y) \nu(dy).\] Then  $\{X_t\}$ is strongly irreducible over $\S \setminus \S_{\eta}$ for all $\eta>0$.
}\end{remark}
\subsection{Robust persistence}

Under an additional assumption, our main condition ensuring persistence is robust to small variations of the model. The importance of this robustness stems from the fact that all models are approximations to reality. Consequently, if nearby
models (e.g. more realistic models) are not persistent despite the focal model being persistent, then one can draw few (if any!) conclusions about the persistence of the modeled biological system. To state our result about robustness, let $g(x,\xi) = g_1(x,\xi),\ldots, g_k(x,\xi)$ be fitness functions. The model
\begin{equation}
\label{eq:discreteperturb}
X_{t+1}=g(X_t,\xi_{t+1})\circ X_t
\end{equation}
is called a \emph{$\delta$-perturbation of (\ref{eq:discrete})} provided $g$ satisfies conditions \textbf{A3}--\textbf{A4} and
\[
\sup_{x\in \S}\E[\|f(x,\xi)-g(x,\xi)\|] = \sup_{x\in \S}\int \|f(x,\xi)-g(x,\xi)\| m(d\xi) \leq \delta.\]
\begin{proposition}
\label{thm:robpersistence}
Assume the dynamics (\ref{eq:discrete}) satisfies hypothesis $(i)$ of Theorem~\ref{thm:persistence} and there exist constants $0<\alpha\leq \beta<\infty$ such that
$\alpha \leq f_i(x,\xi) \leq \beta$ for all $i,x,\xi$. Then there exists $\delta > 0$ such that every $\delta-$perturbation of (\ref{eq:discrete}) satisfies hypothesis $(i).$
\end{proposition}

\section{Continuous time models\label{sec:conts}}

For stochastic differential equation models, we assume, for presentational clarity, that $\S=\Delta=\{x\in \R^k_+: \sum_i x_i =1\}$. However, our results hold more generally for a compact region that is forward invariant for the stochastic dynamics.  The population dynamics on $\Delta$ consists of a ``drift'' term that describes the dynamics in the absence of noise and a ``diffusion'' term that describes the effects of environmental stochasticity on the population dynamics. The drift term for  population $i$ is given by $X^i_tF_i(X_t)$ where $F_i$ is its per-capita growth rate. To allow for correlations of environmental fluctuations across populations, we assume the environmental noise is generated by an $m$-dimensional standard Brownian motion $(B^1_t,\dots,B^m_t)$ and per-capita effect of $B^j_t$ on the growth of population $j$ is given by $\Sigma_i^j(X_t)$. Under these assumptions,  we arrive stochastic differential equations of the form
\begin{equation}
\label{eq:sde}
dX^i_t = X^i_t[F_i(X_t)dt + \sum_{j = 1}^m\Sigma_i^j(X_t)dB^j_t], i = 1, \ldots, k.
\end{equation}
To ensure existence and uniqueness of solutions of \eqref{eq:sde}, we assume that $F_i$ and $\Sigma_i^j$ are real valued Lipschitz continuous maps on $\Delta$. To ensure that the population dynamics remain on $\Delta$ (i.e. $\Delta$ is invariant), we assume that for each $x \in \Delta$, the drift vector
$x\circ F(x)
$
and the diffusion terms
\begin{equation}
\label{defS}
S^j(x) = x\circ\Sigma^j(x), j = 1,\ldots, m,
\end{equation}
are elements of  the tangent space $T\Delta = \{u \in \R^k : \: \sum_{j = 1}^k u_j = 0\}$ of $\Delta$.

The stochastic differential equation (\ref{eq:sde}) defines a continuous time Markov process $\{X_t\}_{t\ge 0}$ on $\Delta.$ We let $\{P_t\}_{t\ge 0}$ denote the associated semigroup defined by
\[P_t h (x) = \E[h(X_t)|X_0 = x]\] for every bounded or nonnegative  measurable function $h: \Delta \to \R.$ $P_th(x)$ is the expected value of $h$ at time $t$ given that initial population state is $x$. 
A probability $\mu$ on $\Delta$ is called {\em invariant} (respectively {\em ergodic}) provided it is invariant (respectively ergodic) for $P_t$ for all $t > 0.$ The \emph{occupation measure} of $\{X_t\}_{t\ge 0}$ is the measure
\[\Pi_t = \frac{1}{t} \int_0^t \delta_{X_s} ds.\]
$\Pi_t(A)$ corresponds to the fraction of time spent in the set $A\subset \Delta$ by time $t$.

The analog of the per-capita growth rate for these continuous time processes is given by
\begin{equation}
\label{eq:newlambda}
\lambda_i(x) = F_i(x) - \frac{1}{2} a_{ii}(x)
\end{equation}
where
\[a_{ij}(x) = \sum_{k = 1}^m \Sigma_i^k(x) \Sigma_j^k(x).\]
When $\lambda_i(x)>0$, the population tends to increase. When $\lambda_i(x)<0$, the population tends to decrease. Like in the discrete-time case, define
\begin{equation}
\label{eq:newlambdamu}
\lambda_i(\mu)  = \int \lambda_i(x) \mu(dx)
\end{equation}
and
\begin{equation}
\label{eq:newlambda*}
\lambda_*(\mu) = \max_i \lambda_i(\mu).
\end{equation}

In Appendix B, we prove the following continuous-time analog of Theorem~\ref{thm:persistence}
\begin{theorem}
\label{thm:persistence2}
Assume that one of the equivalent conditions $(i)$ and $(ii)$ of Theorem \ref{thm:persistence} hold where $\lambda_i(\mu)$ and $\lambda_*(\mu)$ are  given by formulaes (\ref{eq:newlambda}), (\ref{eq:newlambdamu}) and (\ref{eq:newlambda*}). Then the conclusion of Theorem \ref{thm:persistence} hold for the occupation measure of the process $\{X_t\}_{t\ge 0}$ solution to (\ref{eq:sde}). \end{theorem}

To ensure the existence of a unique stationary distribution and convergence toward this distribution, we need an appropriate irreducibility condition that ensures the noise can locally push the dynamics in all directions. More precisely, we call the system (\ref{eq:sde}) {\em nondegenerate} if the column vectors $S^1(x), \ldots , S^m(x)$ span $T\Delta$ for all $x \in \iDelta$.
\begin{theorem}\label{thm:onemore} Assume that (\ref{eq:sde}) is {\em non-degenerate} and the assumption of Theorem \ref{thm:persistence2} holds. Then there exists a unique invariant probability $\pi$ such that $\pi ( \Delta_0)=0.$ Furthermore,
\begin{enumerate}
\item[(i)] The distribution of $X_t$ converges to $\pi$ as $t\to\infty$ whenever $X_0=x\in \iDelta;$ that is
\[\\lim_{t \rightarrow \infty}  \|\P_x[X_t \in \cdot] - \pi\| = 0 \mbox{ for all } x \in \iDelta  ,\]
\item[(ii)] The occupation measures $\Pi_t = \frac{1}{t}\int_0^t  \delta_{X_s}ds$ converge almost surely to $\pi$, whenever $X_0=x\in \iDelta$.
\end{enumerate}
\end{theorem}

\subsection{Robust persistence} Let $\tilde{F}$ and $\tilde{\Sigma}$  be real valued Lipschitz continuous maps on $\Delta$ with the property that for each $x \in \Delta$, the drift vector $x\circ\tilde{F}(x)$
and the diffusion terms
$\tilde{S}^j(x) =x\circ\tilde{\Sigma}^j(x),$ are elements of $T\Delta$.  The stochastic differential equation
 \begin{equation}
\label{eq:sdeperturb}
dX^i_t = X^i_t[\tilde{F}_i(X_t)dt + \sum_{j = 1}^m\tilde{\Sigma}_i^j(X_t)dB^j_t], i = 1, \ldots, k;
\end{equation}
is called a {\em $\delta$-perturbation} of (\ref{eq:sde}) if \[\sup_{x\in \Delta} \|F(x)-\tilde{F}(x)\| + \|\Sigma(x) - \tilde{\Sigma}(x)\| \leq \delta.\]
\begin{proposition}
\label{thm:robpersistence2} Assume  that the dynamics (\ref{eq:sde}) satisfies hypothesis  $(i)$  of Theorem \ref{thm:persistence}. Then there exists $\delta > 0$ such that every {\em $\delta$-perturbation} of (\ref{eq:sde}) satisfies hypothesis $(i)$.
\end{proposition}

In the case $\Sigma(x) = 0,$ Proposition \ref{thm:robpersistence2} combined with Theorem  \ref{thm:persistence2} or \ref{thm:onemore} implies that  every sufficiently small random perturbation
of the deterministic system
\begin{equation}
\label{eq:ode}
\frac{dx}{dt} = x \circ F(x)
\end{equation}
is persistent, provided that
\[ \sup_{\mu}  \int F_i(x) \mu(dx) > 0\] where the supremum is taken over the invariant probabilities of (\ref{eq:ode}) supported by $\Delta_0.$
This fact was also proved in  \citet[Theorem 3.1]{benaim-etal-08}  using other techniques. Hence, Theorems \ref{thm:persistence2}, \ref{thm:onemore} and Proposition \ref{thm:robpersistence2} extend \citet[Theorem 3.1]{benaim-etal-08} beyond  small perturbation of deterministic dynamics. We remark, however, the result obtained in \citet{benaim-etal-08} provides an exponential rate of convergence toward $\pi.$ It would be nice to see whether of not such a rate can be obtained under the more general assumptions of Theorem \ref{thm:onemore}.

\section{Applications\label{sec:apps}}

\subsection{Lottery models and the storage effect} The lottery model of \citet{chesson-warner-81} represents species that require a territory or ``home'' (an area held to the exclusion of others) in order to reproduce. Moreover, the model assumes that space is always in short supply and, consequently, all patches are occupied.  Let $X^i_t$ denote the fraction of space occupied by species $i$ at time $t$. The fraction of adults of species $i$ dying in a time step is $m_i$. The spaces emptied by dying individuals are immediately filled by progeny which are produced at a rate $b_i(X_t,\xi_{t+1})X^i_t$ by species $i$. Here $\xi_t$ is a sequence of i.i.d. random variables that represents environmental stochasticity. If all progeny are equally likely to fill empty spaces, then the probability an empty space is filled by species $i$ equals
\[
\frac{b_i(X_t,\xi_{t+1})X^i_t}{\sum_{j=1}^k b_j(X_t,\xi_{t+1})X^j_t}.
\]
Under these assumptions, the dynamics of the competing species are given by
\begin{equation}\label{eq:lottery}
X^i_{t+1}=(1-m_i)X^i_t+ \sum_j m_j X^j_t \frac{b_i(X_t,\xi_{t+1})X^i_t}{\sum_j b_j(X_t,\xi_{t+1})X^j_t} \quad i=1,\dots, k
\end{equation}
on the state space $\Delta$. For many choices of $b_i$ and $\xi_t$, \eqref{eq:lottery} satisfies the irreducibility assumptions of Theorems~\ref{thm:uniqueness} and \ref{thm:PHC}. For instance, the irreducibility assumptions are satisfied if $b_i(X_t,\xi_{t+1})=\xi^i_{t+1}$ are log-normally distributed or gamma distributed.

When there are two species, \citet{chesson-82} analyzed this model when $b_i(X_t,\xi_t)$ do not depend on $X_t$. We show how our results recover Chesson's persistence criteria. The set of ergodic invariant measures on $\Delta_0$ are the Dirac measures, $\delta_{(0,1)}$ and $\delta_{(1,0)}$, supported on the points    $(0,1)$ and $(1,0)$, respectively. At these ergodic measures, the invasion rates are given by
\begin{center}
\begin{tabular}{c|cc}
$\mu$ & $\lambda_1(\mu)$ & $\lambda_2(\mu)$ \\\hline
$\delta_{(1,0)}$& 0& $\E\left[\log \left( 1-m_2+m_2 \frac{b_2((1,0),\xi)}{b_1((1,0),\xi)}\right)\right]$\\
$\delta_{(0,1)}$& $\E\left[\log \left( 1-m_1+m_1 \frac{b_1((0,1),\xi)}{b_2((0,1),\xi)}\right)\right]$&0
\end{tabular}
\end{center}
Theorem~\ref{thm:persistence} ensures persistence if
\begin{equation}\label{eq:chesson}
\E\left[\log \left( 1-m_1+m_1 \frac{b_1((0,1),\xi)}{b_2((0,1),\xi)}\right)\right]>0
\mbox{ and }
\E\left[\log \left( 1-m_2+m_2 \frac{b_2((1,0),\xi)}{b_1((1,0),\xi)}\right)\right]>0
\end{equation}
Hence, we have recovered the ``mutual invasibility'' condition for persistence of Chesson without making any monotonicity assumptions about the functions $b_i(x,\xi)$ (see also \citet{ellner-84}).

When the per-capita reproductive rates are frequency-independent i.e. $b_i(x,\xi_{t+1})=\xi_{t+1}^i$ and $\xi_t=(\xi_t^1,\xi_t^2)$, the persistence condition \eqref{eq:chesson} can be used to illustrate what \citet{chesson-warner-81} call the ``storage-effect.'' In the absence of environmental stochasticity, \citet{chesson-82} has shown that coexistence is not possible. If there is environmental stochasticity and all individuals die between generations i.e. $m_i=1$ for all $i$, then \eqref{eq:chesson} simplifies to
\[
\E[\log \xi^1]> \E[\log \xi^2] \mbox{ and }\E[\log \xi^2]> \E[\log \xi^1]
\]
Both of these conditions can not be meet in which case \citet[Thm. 3.5]{chesson-82} has shown that one of the species goes extinct with probability one. Hence, when individuals are short lived, coexistence does not occur. On the other hand, if individuals are long-lived i.e. $m_i\approx 0$ for all $i$, then the approximation $\log(1+x)=x+\mathcal{O}(x^2)$ applied to \eqref{eq:chesson} yields the persistence criterion
\[
\E\left[ \frac{\xi^1}{\xi^2}\right]>1 \mbox{ and } \E\left[ \frac{\xi^2}{\xi^1}\right]>1.
\]
This condition is meet when the species exhibit some temporal partitioning: the reproductive rates $\xi_t^1$ and $\xi_t^2$ are not overly correlated. Intuitively, long-lived individuals, unlike short-lived individuals, can ``store'' their numbers over periods of poor conditions and, thereby, take advantage of future good conditions. This ability to store for the future in conjunction with temporal partitioning mediates coexistence.

For lottery models with three or more species, persistence criteria can be more subtle as they may require determining the invasion rates at non-trivial ergodic measures. However, an interesting exception occurs with a rock-paper-scissor version of the lottery model i.e. species $B$ displaces species $A$, $C$ displaces $B$, and $A$ displaces $C$. To model this intransitive interaction, we assume that the per-capita reproductive rates are linear functions of the species frequencies
\[
b_i(X_t,\xi_{t+1})= \sum_j \xi^{ij}_{t+1} X^j_t
\]
where
\[
\xi_t =  \begin{pmatrix} \beta_t & \alpha_t & \gamma_t  \cr \gamma_t & \beta_t & \alpha_t \cr \alpha_t & \gamma_t & \beta_t
\end{pmatrix}
\]
where $\alpha_t>\beta_t>\gamma_t>0$ for all $t$. For simplicity, we assume that $m^i=m$ for all $i$.

For any pair of strategies, say $1$ and $2$, the dominant strategy, $1$ is this case,  displaces the subordinate strategy. Indeed, assume $X^3_0=0$. If $y_t=X^2_t/X^1_t$ and $z_t= \sum_i \xi^i_{t+1}X^i_t$, then
\[
y_{t+1}=\frac{(1-m)z_t+m(\gamma_{t+1} X^1_t+\beta_{t+1} X^2_t)}{(1-m)z_t+m(\beta_{t+1} X^1_t+\alpha_{t+1} X^2_t)}\,y_t<y_t
\]
is a decreasing sequence that converges to $0$. Hence, the only ergodic measures on $\Delta_0$ are Dirac measures $\delta_{x}$ supported on $x=(1,0,0),(0,1,0),(0,0,1)$.  At these ergodic measures, the invasion rates are given by
\begin{center}
\begin{tabular}{c|ccc}
$\mu$ & $\lambda_1(\mu)$ & $\lambda_2(\mu)$ &$\lambda_3(\mu)$ \\\hline
$\delta_{(1,0,0)}$& 0& $\E\left[\log \left( 1-m+m\,\alpha_t/\beta_t\right)\right]$& $\E\left[\log \left( 1-m+m\,\gamma_t/\beta_t\right)\right]$\\
$\delta_{(0,1,0)}$&  $\E\left[\log \left( 1-m+m\,\gamma_t/\beta_t\right)\right]$& 0& $\E\left[\log \left( 1-m+m\,\alpha_t/\beta_t\right)\right]$\\
$\delta_{(0,0,1)}$&$\E\left[\log \left( 1-m+m\,\alpha_t/\beta_t\right)\right]$ &  $\E\left[\log \left( 1-m+m\,\gamma_t/\beta_t\right)\right]$&0
\end{tabular}
\end{center}
A straightforward algebraic competition reveals that the conditions for persistence are satisfied if and only if
\begin{equation}\label{eq:rps}
I(m):=\E\left[\log \left( 1-m+m\,\alpha_t/\beta_t\right)\right]+\E\left[\log \left( 1-m+m\,\gamma_t/\beta_t\right)\right]>0
\end{equation}
We conjecture that if the opposite inequality holds, then persistence does not occur.

To see the role of the storage effect for these rock-paper-scissor communities, we can examine how the sign of $I(m)$ depends on $m$. Since $I(0)=0$ and $I''(m)<0$ for $0\le m\le 1$, $I(m)>0$ for a non-empty interval of $m$ values if and only if
\begin{equation}\label{eq:nec}
I'(0)= \E\left[\frac{\alpha_t}{\beta_t}+\frac{\gamma_t}{\beta_t}\right]-2>0
\end{equation}
Moreover, if
\begin{equation}\label{eq:suf}
I(1)=\E\left[\log \frac{\alpha_t}{\beta_t}+\log \frac{\gamma_t}{\beta_t}\right]>0
\end{equation}
 then the community persists for all $0<m\le 1$. However, if \eqref{eq:nec} holds but \eqref{eq:suf} does not, then the community persists for $0<m<m^*$ for some $m^*<1$. Hence, as in the two species example of Chesson and Warner, competitive communities with intransitives are more likely to coexist if individuals have longer generation times. However, unlike the example of Chesson and Warner,  environmental noise can disrupt as well as enhance coexistence (see discussion in section~\ref{sec:discussion}).

\subsection{Discrete Lotka-Volterra dynamics with disturbances} Consider $k$ interacting species whose dynamics in the absence of environmental disturbances is given by
\[
X_{t+1}=X_{t}\circ \exp(A X_t +b)
\]
where the matrix $A$ describes pairwise interactions between species and $b$ describes the intrinsic rates of growth of each species. These dynamics of this system were studied by \citet{hofbauer-etal-87}. To account for stochastic disturbances of these dynamics, we assume that the fraction of individuals of species $i$ surviving environmental disturbances is $\xi_t^i \in (0,1]$ at time $t$. Then the dynamics become
\begin{equation}\label{LV}
X_{t+1}=X_{t}\circ \exp(A X_t +b) \circ \xi_{t+1}
\end{equation}

An algebraic characterization of boundedness in terms of the matrices $A$ and $b$ remains an open problem even without the stochasticity.  \citet{hofbauer-etal-87} defined the interaction matrix $A$ to be \emph{hierarchically ordered} if there exists a reordering of the indices such that $A_{ii}<0$ for all $i$ and $A_{ij}\le 0$ whenever $i\le j$. While this assumption excludes all types of mutualistic interactions, it allows for any type of predator-prey or competitive interaction. The following lemma extends work of \citet{hofbauer-etal-87} by showing that hierarchically ordered systems remain bounded in the presence of environmental disturbances. For these systems, the irreducibility conditions of Theorems~\ref{thm:uniqueness} and \ref{thm:PHC} are meet whenever $\xi_t$ has a positive, continuous density on the interval $(0,1)$. 

\begin{lemma} If \eqref{LV} is hierarchically ordered, then there exists $K>0$ such that $X_t\in [0,K]^k$ for $t\ge k+1$. \end{lemma}

\begin{proof} Following \citet{hofbauer-etal-87} observe that
\[
X_{t+1}^1 \le X_t^1 \exp(A_{11} X_t^1+b_1)
\]
for all $t$ as $A_{1j}\le 0$ for all $j\ge 2$ and $\xi_t^1\le 1$. Hence, $X_{t+1}^1$ is bounded above by $K_1=-\exp(b_1-1)/A_{11}$, and  $X_t^1\in [0,K_1]$ for  $t\ge 2$.

Assume that there exist $K_i$ such that $X_t^i \in [0,K_i]$ for $i\le j-1$ and $t\ge i+1$. We will show that there exists $K_{j}$ such that $X_t^{j}\in [0,K_{j}]$ for $t\ge j+1$. Indeed, by the hierarchically ordered assumption and our inductive assumption,
\[
X_{t+1}^{j}\le X_t^{j}\exp(A_{jj}X^{j}_t+b_{j}+\sum_{i<j} |A_{ji}|K_j)
\]
for $t\ge j$. Hence, $X_t^j \le K_j$ for $t\ge j+1$ where
 $K_j=-\exp(b_{j}+\sum_{i<j} |A_{ji}|K_j-1)/A_{jj}$. Defining $K=\max K_j$ completes the proof.
\end{proof}

The following lemma shows that verifying persistence can reduce to a linear algebra problem. In particular, this lemma implies that all the permanence criteria developed by \citet{hofbauer-etal-87} for hierarchal systems extends to these stochastically perturbed Lotka-Volterra systems.

\begin{lemma} Let $\mu$ be an ergodic measure for \eqref{LV} and $I\subset\{1,\dots,k\}$ be such that $\mu(\{x\in \S: x^i>0$ iff $i\in I\}=1$. Define $\beta_i = b_i+\E[\log \xi^i_t]$.  If there exists a unique solution $\hat x$ to
\[
\sum_j  A_{ij}\hat x_j+\beta_i=0 \mbox{ for }i\in I\mbox{ and }
\hat x_i=0 \mbox{ for } i\notin I
\]
 then
\[
\lambda_i(\mu)=\left\{\begin{array}{cc} 0& \mbox{ if }i\in I\\ \sum_j A_{ij} \hat x_j +\beta_i & \mbox{otherwise.}
\end{array}\right.
\]
\end{lemma}

\begin{proof}
Let $\mu$ and $I$ be as assumed in the statement of the lemma. Assertion $(iii)$ of Proposition~\ref{prop:invasion}  implies that
\[
0=\lambda_i(\mu)= \sum_j  A_{ij} \int x_j\,\mu(dx)+\beta_i
\]for all $i\in I$. Since we have assumed there is a unique solution to this system of linear equations, it follows that $\int x_i \mu(dx)=\widehat x_i$ for all $i$ and the lemma follows.
 \end{proof}

\subsection{Stochastic replicator dynamics} A particular case of the continuous time equations \eqref{eq:sde} is given by  the {\em stochastic replicator dynamics} introduced by \cite{fudenberg-harris-92}. Assume that the {\em fitness} of population $i$ is described by a function $f_i : \Delta \mapsto \R$ of the state variable $x$ and that the per capita growth rate of the number of individuals in population $i$ is stochastic, given by  sum of the fitness of strategy $i$ and a standard Brownian motion $B_i(t)$:
\begin{equation}\label{eq:stgrowth}
dY_t^i = Y^i_t \left(f_i(X_t) + \sigma_i dB_t^i\right),
\end{equation}
where $X^i_t = Y^i_t/ \sum_j Y^j_t$ and $\sigma_i > 0$. Then the law of motion for the state $X_t$ can be obtained via a straightforward application of Ito's formula and takes the form (\ref{eq:sde})
with
\[F_i(x) = f_i(x) - \sigma_i^2 x_i - \sum_j x_j(f_j(x) - \sigma_j^2 x_j)\]
and
\[\Sigma_i^j(x) = (\delta_{ij} -x_j)\sigma_j.\]

If there are only two types (i.e. $k=2$), then the only ergodic measures on $ \Delta_0$ are the Dirac measures at $x=(1,0)$ and $x=(0,1)$, respectively. At these Dirac measures, the invasion rates are given by
\begin{center}
\begin{tabular}{c|cc}
$\mu$ & $\lambda_1(\mu)$ & $\lambda_2(\mu)$ \\\hline
$\delta_{(1,0)}$& $0$&$ f_2(1,0)-f_1(1,0)-\frac{1}{2}(\sigma_2^2-\sigma_1^2)$\\
$\delta_{(0,1)}$& $ f_1(0,1)-f_2(0,1)-\frac{1}{2}(\sigma_1^2-\sigma_2^2)$& $0$
\end{tabular}
\end{center}
Hence, both strategies persist if
\[
f_1(0,1)-\frac{\sigma_1^2}{2}>f_2(0,1)-\frac{\sigma_2^2}{2}
\]
and
\[
f_2(1,0)-\frac{\sigma_2^2}{2}>f_1(1,0)-\frac{\sigma_1^2}{2}
\]
When these inequalities are satisfied, one can solve explicitly for the density of the positive stationary distribution of $X^1=x$ on $[0,1]$  (see, e.g., \cite{kimura-64})
\[
\rho(x)=\frac{C}{V(x)} \exp\left(2 \int \frac{F_1(x,1-x,0)}{V(x)} dx\right)
\]
where $V(x)=x^2(1-x)^2(\sigma_1^2+\sigma_2^2)$ and $C$ is a normalization constant. For example, if $f_i$ are linear functions, then this stationary distribution is given by a beta distribution as we illustrate in the next example.

Since we can solve for non-trivial ergodic measures for two interacting types, we can derive explicit conditions for persistence of three interacting types. As an illustration,  consider three interacting types with per-capita growth rates $f_1(x)=r_1+b\,x_3$, $f_2(x)=r_2$, and $f_3(x)=r_3-c\,x_1$. Here, interactions between types $1$ and $3$ provide a benefit $b>0$ to type $1$ and a cost $c>0$ to type $3$. To allow for coexistence, we assume the following tradeoff $r_3-\sigma_3^2/2>r_2-\sigma_2^2/2>r_1-\sigma_1^2/2$.

Our analysis begins with pair-wise interactions.  When  type $1$ is not present i.e. $Y^1_0=0$,  the remaining types $i=2,3$ exhibit geometric Brownian motions of the form $Y_t^i=Y_0^i \exp\left( (r_i -\sigma_i^2/2)t +\sigma_i B_t^i\right)$ where $B_t^i$ are independent Brownian motions. Since $r_3-\sigma_3^2/2>r_2-\sigma_2^2/2$,  $X_t$ converges almost surely to $(0,0,1)$  whenever $Y_0^1=0$ and $Y_0^3>0$. Similarly,  when type $3$ isn't present i.e. $Y_0^3=0$, the remaining types $i=1,2$ exhibit geometric Brownian motions and $X_t$ converges almost surely to $(0,1,0)$ whenever $Y_0^2>0$. To determine the outcome of the pairwise interaction between genotypes $1$ and $3$, we need the invasion rates
\[
\lambda_1(0,0,1)=r_1+b -\sigma_1^2/2 -r_3+\sigma_3^2/2
\mbox{ and } \lambda_3(1,0,0)=r_3-c -\sigma_3^2/2 -r_1+\sigma_1^2/2
\]
Both of these invasion rates are positive provided that
\begin{equation}\label{eq:coexist}
b>r_3 -\sigma_3^2/2 -r_1+\sigma_1^2/2>c
\end{equation}
When \eqref{eq:coexist} holds, there is a unique invariant measure $\mu$ on  $\{x\in \Delta: x_1x_3>0,x_2=0\}$  whose density $\rho(x_1):=\rho(x_1,1-x_1)$ is given by
\[
\rho(x_1)= \frac{C}{x_1^2(1-x_1^2)(\sigma_1^2+\sigma_2^2)} 
\exp\left( 2\int \frac{r_1+b(1-x_1)-r_3+cx_1+\sigma_2^2 (1-x_1)-\sigma_1^2 x_1}{x_1^2(1-x_1^2)(\sigma_1^2+\sigma_2^2)}dx\right) 
\]
which upon integration yields
\begin{equation}\label{eq:density}
\rho(x_1)=\frac{x_1^{\alpha-1}x_3^{\beta-1}}{B(\alpha,\beta)}
\end{equation}
where $B(\alpha, \beta)$ is a normalization constant and 
\[
\alpha=\frac{2(\sigma_3^2-r_3+r_1+b)}{\sigma_1^2+\sigma_3^2}-1
\qquad
\beta=\frac{2(\sigma_1^2+r_3-r_1-c)}{\sigma_1^2+\sigma_3^2}-1.
\]
\citet[Proposition 1]{fudenberg-harris-92} provide a detailed derivation of this stationary distribution for linear $f_1$ and $f_3$.

To understand the fate of the three interacting genotypes, there are (generically) three cases to consider. First, assume that \eqref{eq:coexist} is satisfied. The invasion rate for type $2$ at the invariant measure $\mu$, see \eqref{eq:density}, is given by
\begin{equation}\label{eq:invade}
\lambda_2(\mu)=\frac{b\sigma_3^2-(b-c)\sigma_2^2-c\sigma_1^2-2br_3+2(b-c)r_2+2cr_1+2bc}{2(b-c)}
\end{equation}
Whenever $\lambda_2(\mu)>0$, Theorem 4 ensures there is a unique positive stationary distribution on $\Delta$ by choosing $p_3\gg p_2 \gg p_1>0$.  Since \eqref{eq:coexist} implies $b-c>0$, \eqref{eq:invade} implies that stochastic fluctuations in genotype $3$'s per-capita growth rate can mediate coexistence, while stochastic fluctuations in the per-capita growth rates of the other two genotypes can disrupt coexistence.

Next, assume that \eqref{eq:coexist} doesn't hold. If $b<r_3-\sigma_3^2/2-r_2+\sigma_2^2/2$, then the invasion rates $\lambda_1(0,0,1)$ and $\lambda_2(0,0,1)$ are both negative and we conjecture coexistence doesn't occur. Alternatively if $b,c>r_3-\sigma_3^2/2-r_2+\sigma_2^2/2$, then the boundary dynamics exhibit a rock-paper-scissor dynamic and the only ergodic invariant measures are the Dirac measures at the vertices.  At these ergodic measures, the invasion rates are given by
\begin{center}
\begin{tabular}{c|ccc}
$\mu$ & $\lambda_1(\mu)$ & $\lambda_2(\mu)$ &$\lambda_3(\mu)$ \\\hline
$\delta_{(1,0,0)}$& $0$&$r_2-\sigma_2^2/2-r_1+\sigma_1^2/2$& $r_3-c-\sigma_3^2/2-r_2+\sigma_2^2/2$\\
$\delta_{(0,1,0)}$&  $r_1-\sigma_1^2/2-r_2+\sigma_2^2/2$& 0&  $r_3-\sigma_3^2/2-r_2+\sigma_2^2/2$\\
$\delta_{(0,0,1)}$& $r_1+b-\sigma_1^2/2-r_3+\sigma_3^2/2$&   $r_2-\sigma_2^2/2-r_3+\sigma_3^2/2$&0
\end{tabular}
\end{center}
A standard computation yields that the persistence criterion is satisfied when product of the positive invasion rates is greater than the product of the absolute value of the negative invasion rates. This occurs when $b>c$. Hence, for this rock-paper-scissor dynamic, environmental stochasticity has no effect on coexistence.

\section{Discussion\label{sec:discussion}}

Understanding under what conditions interacting populations, whether they be plants, animals, or viral particles, coexist is a question of theoretical and practical importance in population biology. Both biotic interactions and environmental fluctuations are key factors that facilitate or disrupt coexistence. To better understand this interplay between these deterministic and stochastic forces, we develop a mathematical theory  extending the nonlinear theory of permanence for deterministic systems to randomly forced nonlinear systems. This theory provides a biologically interpretable criterion for coexistence in the sense of stochastic boundedness~\citep{chesson-78,chesson-82}. Using this theory, we illustrate that environmental noise enhances or inhibits coexistence in communities with rock-paper-scissor dynamics, has no effect on coexistence in certain Lotka-Volterra communities, and can promote or inhibit genetic diversity.

Our condition for coexistence requires that there is a fixed set of weights associated with the interacting populations and this weighted combination of populations' invasion rates is positive for any (ergodic) stationary distribution associated with a subcollection of populations. This criterion is the stochastic analog of a permanence criterion for deterministic systems~\citep{hofbauer-81,jde-00,garay-hofbauer-03}. Since these invasion rates, defined as the average per-capita growth rates on the stationary distribution, equal zero for populations supported by the stationary distribution, this criterion requires that a missing population has a positive invasion rate. Hence, for pair-wise interactions, this criterion reduces to the ``mutual invasibility'' criterion.  When this condition holds and there is sufficient noise in the system (i.e. irreducible), we have shown the populations approach a unique positive stationary distribution whenever all types are initially present. Hence, the probability that the abundance of any population falls below a critical threshold is arbitrarily small for sufficiently small thresholds. Moreover, the fraction of time any population spends below this threshold is arbitrarily small for sufficiently small thresholds.

The need for this generalization of the mutual invasiblity criterion is illustrated in the deterministic literature by models of communities exhibiting rock-paper-scissor type dynamics~\citep{hofbauer-sigmund-98}. Here, we extended this analysis to stochastic counterparts of these models. If we assume that dominant strategies (e.g. rock) in these models gain a benefit $b_t$ when playing subordinate strategies (e.g. scissor) and subordinate strategies pay a cost $c_t$ when playing dominant strategies, then  our coexistence condition \eqref{eq:nec} for long-lived individuals becomes
\begin{equation}\label{disc}
\E\left[ \frac{b_t}{\beta_t}\right] > \E\left[ \frac{c_t}{\beta_t}\right]
\end{equation}
where $\beta_t$ is the ``base'' payoff. The effect of stochasticity in $b_t$, $c_t$, and $\beta_t$ on whether this criterion holds depends on the correlations between the various payoffs. Negative correlations between base payoffs and benefits (i.e. getting large benefits when base payoffs are small)  makes \eqref{disc} more likely to hold, while negative correlations between base payoffs and costs make it less likely to hold. Hence, stochasticity can facilitate coexistence when there are negative correlations between benefits and base-payoffs, but inhibit coexistence when there are positive correlations between benefits and base-payoffs.

For three competing genotypes in which the genotypes with the highest per-capita growth rate is exploited by the genotype with the lowest per-capita growth rate, we have shown that the effect of environmental fluctuations on coexistence is subtle. When the three genotypes exhibit a rock-paper-scissor dynamic, stochastic fluctuations have no effect on the coexistence criterion; coexistence requires that the benefit to the exploiter exceed the cost paid by the exploited. When there is no rock-paper-scissor dynamic, fluctuations in the per-capita growth rate of the exploited genotype can enhance diversity, while fluctuations in the other two genotypes can disrupt coexistence. Since this noise-induced coexistence occurs in populations with overlapping generations (i.e. a stochastic differential equation model), these results partially support \citet{ellner-sasaki-96}'s assertion that ``fluctuating selection can readily maintain genetic variance in species where generations overlap in such a way that only a fraction of the population is exposed to selection.''

We also have shown that stochastic variation in mortality or disturbance rates have no effect the coexistence criteria for discrete-time Lotka Volterra models developed by \citet{hofbauer-etal-87}. This surprising outcome stems from the fact that the per-capita growth rates in these models are linear functions in the population abundances. Adding non-linearities (e.g. predator saturation) to the per-capita growth rates will alter this conclusion, but the nature of this alteration remains to be understood.

Numerous mathematical challenges remain at this interface of random forcing and biotic interactions. For example, do the same criteria hold when there are temporal correlations in the environmental variables? We suspect the answer is yes. Alternatively, we conjecture that inverting the coexistence criterion (i.e. a convex combination of invasion is negative for all stationary distributions supporting subsets of species) implies an asymptotic approach to extinction with probability one. While this conjecture has been proven for stochastic differential equations with small diffusion terms~\citep{benaim-etal-08}, it needs to be shown for stochastic difference equations or stochastic differential equations with large diffusion terms. Finally, only recently have invasion-based permanence  criteria  been developed for deterministic models of  structured interacting populations~\citep{jde-10}. These structured models can account for heterogeneity amongst individuals in terms of the location, size, and age. Developing a mathematical framework to deal with these heterogeneities is an exciting challenge that would help us understand how  interactions between individual heterogeneity, temporal heterogeneity, and and biotic interactions determine diversity.

\paragraph{Acknowledgements.} SJS was supported by United States National Science Foundation Grant  DMS-0517987  and MB was supported by Swiss National Foundation  Grant 200021-103625/1.

\bibliography{../../seb}

\appendix

\section{Proofs for discrete time models}

\subsection{Proof of Proposition~\ref{prop:invasion}}
\begin{lemma}
\label{th:birk2}
Let $g : \S \times E \mapsto \R$ be a measurable map such that $\sup_{x\in C} \int g(x,\xi)^2 m(d\xi)<\infty. $ Define $\bar{g}(x) = \int g(x,\xi) m(d\xi).$ Then
\begin{itemize}
\item[(i)] For all $x \in \S$ and $X_0 = x$
\[\lim_{t \to  \infty} \frac{\sum_{s = 0}^{t-1} g(X_s,\xi_{s+1}) -
 \sum_{s = 0}^{t-1} \bar{g}(X_s)}{t} = 0.\]
with probability one.

\item[(ii)] Let $\mu$ be an  invariant (respectively ergodic) probability measure for $(X_t)$, then there exists a bounded measurable map $\hat{g}$ such that with probability one and for $\mu$-almost every $x$ \[\lim_{t \to  \infty} \frac{\sum_{s = 0}^{t-1} g(X_s,\xi_{s+1})}{t} = \lim_{t \to \infty}  \frac{\sum_{s = 0}^{t-1} \bar{g}(X_s)}{t} = \hat{g}(x)
\mbox{ when }X_0=x.\]
 Furthermore $$\int \bar{g}(x) \mu(dx) = \int \hat{g}(x) \mu(dx) \mbox{ (respectively }
\hat{g}(x) = \int \bar{g}(x) \mu(dx) \quad \mu-\mbox{almost surely} ).$$
\end{itemize}
\end{lemma}
\begin{proof} The first assertion follows from the strong law of large number for martingales, since
$g(X_s,\xi_{s+1}) - Pg(X_s)$ is a square integrable martingale difference.
The second assertion follows from  Birkhoff's ergodic theorem applied to stationary Markov Chains (see \cite{meyn-tweedie-93}, Theorem 17.1.2)
\end{proof}
The first two assertions of Proposition~\ref{prop:invasion}  follow directly from the preceding  lemma applied to $g(x,\xi) = \log f_i(x,\xi)$. For the third assertion, notice that,  by assertion $(i)$ of Proposition~\ref{prop:invasion}
\[\lim_{t \rightarrow \infty} \frac{\log X^i_t}{t} = \hat{\lambda}_{i}(x)\] for $\mu$-almost all $x \in \iS.$
Let $\S^{i,\eta} = \{x \in \S \: : x_i \geq \eta\}$ and $\eta^*>0$ be such that $\mu(\S^{i,\eta})>0$ for all $\eta\le \eta^*$. By Poincar\'e Recurrence Theorem, for $\mu$ almost all $x \in \S^{i,\eta}$
\[\P_x[ X_t \in \S^{i,\eta} \mbox{ infinitly often } ] = 1\]
for $\eta\le\eta^*$.
Thus $\hat{\lambda}_{i}(x) = 0$ for $\mu$-almost all $x \in \S^{i,\eta}$ with $\eta\le \eta^*$.
Hence $\hat{\lambda}^{i}(x) = 0$ for $\mu$-almost all $x \in \bigcup_{n \in \mathbf{N}} \S^{i,1/n} = \{x \in \S \: : x_i > 0\}.$ This proves assertion $(iii).$

\subsection{Proof of Theorem~\ref{thm:persistence}}
The proof of the first assertion of the theorem follows from the following lemma.
\begin{lemma}
\label{hyp}
The following two conditions are equivalent:
\begin{enumerate}
\item[(i)] For all invariant probability measures $\mu$  supported on $\S_0$,
\[
\lambda_*(\mu) := \max_i  \lambda_i(\mu)>0
\]
\item[(ii)]
There exists $p\in \Delta$ such that
\[
\sum_i p_i \lambda_i (\mu) >0
\]
for all ergodic probability measures $\mu$  supported by $\S_0.$
\end{enumerate}
\end{lemma}

\begin{proof}To see the equivalence of the conditions we need the following version of the minimax theorem (see, e.g., \cite{simmons-98}):

\begin{theorem}[Minimax theorem] Let $A,B$ be Hausdorff topological vector spaces and let
$\Gamma : A \times B \to \R$ be a continuous bilinear function. Finally, let $E$ and $F$ be nonempty,
convex, compact subsets of $A$ and $B$, respectively. Then
\[
\min_{a\in E}\max_{b\in F}\Gamma(a, b) = \max_{b\in F}\min_{a\in E} \Gamma(a,b)
\]
\end{theorem}

We have that
\[
\min_\mu \max_i \lambda_i (\mu) = \min_\mu\max_{p\in \Delta} \sum_i p_i \lambda_i(\mu)
\]
where the minimum is taken over invariant probability measures $\mu$ with support in $\S_0$.
Define $A$ to be the dual space to the space bounded continuous functions from $\S_0$ to $\R$ and
define $B=\R^k$. Let $D\subset A$ be the set of invariant probability measures and
$E=\Delta$. With these choices, the Minimax theorem  implies that
\begin{equation}\label{eq:minmax}
\min_\mu \max_i \lambda_i (\mu) = \max_{p\in \Delta} \min_\mu \sum_i p_i \lambda_i(\mu)
\end{equation}
where the minimum is taken over invariant probability measures $\mu$ with support in $\S_0$. By the ergodic decomposition theorem~\cite{mane-83}, the minimum of the right hand side of \eqref{eq:minmax}
is attained at an ergodic probability measure with support in $\S_0$. Thus, the equivalence of the conditions
is established. \end{proof}

The proof of the second assertion of the theorem follows from the next two lemmas.

\begin{lemma} For all $\epsilon>0$, there exists a $\eta>0$ such that
$\mu(\S_\eta)\le \epsilon$
for every invariant probability measure $\mu$ with $\mu(\iS)=1$.
\end{lemma}

\begin{proof} Suppose to the contrary, there exists $\epsilon>0$ and invariant measures $\mu_n$ such that $\mu_n(\iS)=1$ and $\mu(\S_{1/n})>\epsilon$ for all $n\ge 1$. By Proposition~\ref{prop:invasion}, $\lambda_*(\mu_n)=0$ for all $n$. Let $\mu$ be a weak* limit point of these measures. Then $\mu(\S_0)\ge \epsilon$ and $\lambda_*(\mu)=0$. Since $\S=\S_0\cup \iS$ and $\S_0$, $\iS$ are invariant, there exists $\alpha>0$ such that $\mu=\alpha \nu_0+(1-\alpha)\nu_1$ where $\nu_i$ are invariant measures satisfying $\nu_0(\S_0)=1$ and $\nu_1(\iS)=1$. By Proposition~\ref{prop:invasion}, $\lambda_*(\nu_1)=0$. By assumption, $\lambda_*(\nu_0)>0$. Hence, $0=\lambda_*(\mu)\ge \alpha \lambda_*(\nu_0)>0$, a contradiction.

\end{proof}

\begin{lemma}\label{lem:pi} For all $x\in \iS$, with probability one the set of weak* limit points  of $\Pi_t$ is a nonempty compact set consisting of invariant probabilities $\mu$  such that  $\mu(\iS)=1$.
\end{lemma}

\begin{proof}
 The process $\{X_t\}_{t=0}^\infty$ being a Feller Markov chain over a compact set $\S$, the set of weak* limit points of $\{\Pi_t\}_{t=0}^\infty$ is almost surely a non-empty compact subset of $\PS$ consisting of invariant probabilities. To see why this latter point is true, let $h:\S\to\R$ be continuous function and define
\[
g(x,\xi)=h(x\circ f(x,\xi)) \mbox{ and } \bar g(x) =\int_\S g(x,\xi)\,m(d\xi).
\] 
Since $X_{t+1}=X_t\circ f(X_t,\xi_{t+1})$, $h(X_{t+1})=g(X_t,\xi_{t+1})$ and 
\begin{eqnarray*}
\lim_{t\to\infty} \int_\S h(x) \,\Pi_t(dx) - \int_\S Ph(x) \,\Pi_t(dx)&= \lim_{t\to\infty}&\frac{1}{t}\left(\sum_{s=1}^{t}h(X_s)-\int_\S h(f(X_s,\xi)\circ X_s) \, m(d\xi)\right)\\
&=& \lim_{t\to\infty}\frac{1}{t} \left(\sum_{s=0}^{t-1} g(X_{s},\xi_{s+1})-\bar g(X_s) \right)+\frac{1}{t}\left( \bar g(X_t)-\bar g(X_0)\right)\\
&=& 0 \mbox{ almost surely.}
\end{eqnarray*}
where the last line follows from assertion (i) of Lemma \ref{th:birk2}.
Hence, $\int_\S h(x) \, \mu(dx)= \int_\S Ph(x)\,\mu(dx)$ with probability one for  weak* limit points $\mu$ of $\Pi_t$. Since $\S$ is compact, the set of continuous functions from $\S$ to $\R$ is separable metric space and with probability one $\int_\S h(x) \, \mu(dx) = \int_\S Ph(x)\,\mu(dx)$ for all weak* limit points of $\Pi_t$ and all continuous functions $h:\S\to\R$. Thus, the weak* limit points of $\Pi_t$ are almost-surely invariant probability measures.

Assertion $(i)$ of  Lemma \ref{th:birk2} applied to $g(x,\xi) = \log(f_i(x,\xi))$ gives
we have
\begin{eqnarray*}
\lim_{t \to \infty} \frac{\log X^i_t-\log x_i  - \sum_{s=0}^{t-1} \lambda_i(X_s )}{t} = 0.
\end{eqnarray*}
 Since $\limsup_{t\to\infty}\frac{1}{t} \left(\log X^i_t-\log x_i\right) \le 0$ almost surely, we get that
 \begin{equation}
 \label{eq:noname}
 \lambda_*(\mu) \le 0
 \end{equation}
 almost surely for any weak* limit point $\mu$ of $\{\Pi_t\}_{t=0}^\infty$.

 Since $\S=\iS\cup \S_0$, $\iS$ is invariant, and $\S_0$ is invariant, there exists $\alpha\in (0,1]$ such that $\mu=(1-\alpha)\nu_0+\alpha \nu_1$ where $\nu_0$ is an invariant probability measure with $\nu_0( \S_0)=1$ and $\nu_1$ is an invariant probability measure with $\nu_1(\iS)=1$. By Proposition \ref{prop:invasion},  $\lambda_i(\nu_1) = 0$ for all $i.$ Thus,  $(1-\alpha) \lambda_i(\nu_0) \leq 0$ for all $i$. Since by assumption $\lambda_i(\nu_0) > 0$ for some $i,$ $\alpha$ must be $1.$

\end{proof}

\subsection{Proof of Theorems~\ref{thm:uniqueness} and \ref{thm:PHC}}
\begin{lemma}
\label{th:lem2}
There exists $\eta>0$ and $\epsilon > 0$ such that
\begin{enumerate}
\item[(i)] $\lambda_*(\mu) \geq \epsilon$ for every invariant probability $\mu$  with $\mu(\S_\eta) = 1$, and
\item[(ii)] $\P_x[X_t\notin \S_\eta$ for some $t]=1$ for all $x\in \iS$.
\end{enumerate}
\end{lemma}

\begin{proof}
To prove $(i)$, assume to the contrary that there exists a sequence $\{\mu_n\}_{n=1}^\infty$ of invariant probabilities such that $\lambda_*(\mu_n) \leq 1/n$
and $\mu_n(\S_{1/n}) = 1.$
Let $\mu$ be a weak* limit point of the $\{\mu_n\}_{n=1}^\infty$. Hence $\lambda_*(\mu) = 0$ by continuity of $\lambda_*,$ and $\mu(\S_0) = 1$ since $\mu_n(\S_a) = 1$ for all $a > 0$ and $n$ large enough. However, this contradicts the assumption that $\lambda_*(\mu)>0$. Hence, there exists $\epsilon>0$ and $\eta>0$ such that $(i)$ holds.

To prove $(ii)$, let $\mathcal{E}$ be the event $\mathcal{E} = \{\forall t \geq 0: X_t \in \S_\eta\}.$ On $\mathcal{E},$ $\Pi_t$ is almost surely  supported by $\S_\eta$. Hence, by $(i)$, $\lambda_*(\mu) \geq \epsilon$ almost surely on $\mathcal{E}$  for any weak* limit point $\mu$ of $\Pi_t$. This contradicts (\ref{eq:noname}) in the proof of Lemma~\ref{lem:pi}. Hence,  $\mathcal{E}$ has probability zero,
\end{proof}

We now pass to the proof of Theorem \ref{thm:uniqueness}.
Let $\eta > 0$ be like in Lemma \ref{th:lem2} (ii) and $\Phi$ the probability on $\S \setminus \S_{\eta}$ given by the irreducibility assumption.  Then, for all $x \in \iS$ and every Borel set $A \subset \S_{\eta}$ 
$\P_x[\exists n \geq 1 \: X_n \in A]> 0$ whenever $\Phi(A) > 0.$ In other words, $\{X_t\}$ is a {\em $\Phi-$irreducible Markov chain} on $\iS$ in the sense of \citet[Chapter 4, Section 4.2]{meyn-tweedie-93}. It then follows  (see \citet[Proposition 10.1.1, Theorem 10.4.4]{meyn-tweedie-93}) that $\{X_t\}$ admits at most one invariant probability measure on $\iS$ and  Theorem \ref{thm:uniqueness} follows from Lemma \ref{lem:pi}.
  
If one now assume that $\{X_t\}$ is strongly irreducible over $\S \setminus \S_{\eta}$ then $\{X_t\}$ becomes Harris recurrent and aperiodic on $\iS.$ Since, by Theorem  \ref{thm:uniqueness} it is a positive Harris chain,  Theorem \ref{thm:PHC} follows from Orey's theorem (see \citet[Theorem 18.1.2]{meyn-tweedie-93})

\subsection{Proof of Proposition \ref{thm:robpersistence}}
Assume to the contrary that there exists a sequence of fitness maps $\{g_n=(g_n^1,\dots,g_n^k)\}_{n=1}^\infty$ satisfying assumptions \textbf{A3}--\textbf{A4} such that
\begin{equation}
\label{eq:gnf}
\lim_{n \rightarrow \infty} \sup_{x\in\S} \E[\| g_n(x,\xi) - f(x,\xi)\|] = 0
\end{equation}
and
\[\max_{i} \int \log(g_n^i(x,\xi)) m(d\xi) \mu_n(dx) \leq 0\]
where $\mu_n$ is an invariant measure, supported by $\S_0,$ for the operator $P_n$ associated to the Markov chain
\[X_{t+1} = g_n(X_t, \xi_{t+1}) \circ X_t.\]
By compactness of $\S$ we may assume that $\mu_n\rightarrow \mu$ in the weak* topology.  Since $\alpha \leq f_i \leq \beta$ for all $i$ and (\ref{eq:gnf}) holds, it follows  that
\[
\lim_{n\to\infty} \int \log(g_n^i(x,\xi)) m(d\xi) \mu_n(dx) =\int \log(f_i(x,\xi)) m(d\xi) \mu(dx).\]
Hence, $\lambda_i(\mu)\le 0$ for all $i$.
It remains to prove that $\mu$ is invariant for $P$ to reach a contradiction.
Let $h : \S \mapsto \R$ be a continuous map.
Let $\epsilon > 0.$ By uniform continuity there exists $\delta > 0$  such that for all $x,u,v \in \S,$  \[\|u-v\| \leq \delta \Rightarrow |h(x \circ u)-h(x \circ v)| \leq \epsilon.\]Thus
\begin{eqnarray*}|P_nh(x) - Ph(x)| &=& |\E[ h(x \circ g_n(x,\xi)) - h(x \circ f(x,\xi))]| \\
&\leq& 2 \|h\| \P[ \|g_n(x,\xi)-f(x,\xi)\| \geq \delta] + \epsilon\\
&\leq & 2 \|h\|  \frac{\E[\|g_n(x,\xi)-f(x,\xi)\|]}{\delta} + \epsilon.
\end{eqnarray*}
It then follows from (\ref{eq:gnf}) that $\lim_{n\to\infty}P_nh(x) =Ph(x)$ uniformly in $x$. Therefore
\[\lim_{n\to\infty}\int P_n h (x) \mu_n(dx) = \int P h (x) \mu(dx).\]
Since by invariance of $\mu_n$ for $P_n$
\[
\lim_{n\to\infty}\int P_n h (x) \mu_n(dx) = \lim_{n\to\infty}\int h(x) \mu_n(dx)  =\lim_{n\to\infty} \int h(x) \mu(dx)\]
 we get
\[\int Ph(x) \mu(dx) = \int h(x) \mu(dx),\] proving that $\mu$ is invariant for $P.$
\section{Proofs for the continuous time models}

\subsection{Proof of Theorem \ref{thm:persistence2}}
 Let $\cL$ be the  infinitesimal generator of $\{X_t\}_{t\ge 0}.$ It
 acts on $C^2$ functions according to the formula
\begin{equation}
\label{eq:defL}
\cL \psi(x) = \sum_i \frac{\partial \psi}{\partial x_i}(x) x_i F_i(x)  + \cA\psi(x)
\end{equation}
where
\begin{equation}
\label{eq:defA}
\cA\psi(x) = \frac{1}{2}\sum_{i,j} x_i x_j a_{ij}(x) \frac{\partial^2 \psi}{\partial x_i x_j}(x)
\end{equation}
By Ito's formulae
\[\psi(X_t) - \psi(x) - \int_0^t \cL \psi(X_s) ds = M_t\] is a martingale given by $M_0 = 0$ and
\[dM_t = \sum_{i= 1}^k  \frac{\partial}{\partial x_i} (X_t) \sum_{j= 1}^m S_i^j(X_t) dB_t^j\]
where $S^j$ is the vector given by (\ref{defS}).
Applying this to $\psi(x) = \log(x_i)$
gives
\[\log(X_t^i) - \log(x^i) - \int_0^t \lambda_i(X_s) ds = M_t\]
with \[dM_t = \sum_{j = 1}^m \Sigma_i^j(X_t) dB_t^j.\]
Hence \[d \langle M \rangle_t = \sum_{j = 1}^m ((\Sigma_i^j(X_t))^2dt\] so that
\[\langle M \rangle_t \leq C t.\] Thus,  by the strong law of large numbers for martingales,
\[\lim_{t \to \infty} \frac{\log(X_t^i) - \log(x^i) - \int_0^t \lambda_i(X_s) ds}{t} = 0\] almost surely.
 The end of the proof is like  the proof of Theorem \ref{thm:persistence}. Details are left to the reader.

\subsection{Proof or Theorem \ref{thm:onemore}}
By the nondegeneracy assumption, there exists (see e.g~\cite{durrett-96b}, Theorem 3.8, Chapter 7) a continuous positive kernel $p_t(x,y)$ such that
\[P_t \psi(x) = \int p_t(x,y)\psi(y) dy.\]
Therefore, Theorem \ref{thm:persistence} applies to $P_t$ for any $t > 0.$

Let $\pi_t$ denote the unique positive invariant probability measure of $P_t$ for $t > 0.$ We claim that $\pi_t$ is independent of $t.$  Indeed,  $\pi_t$ is invariant for $P_{kt} = P_t^k$ for all $t > 0$ and $k \in \mathbb{N}.$
It follows that $\pi_{k/2^n}$  is independent of $k$ and $n$, and so, by the density of the dyadic rational numbers in the reals, $\pi_t = \pi$ for all $t > 0.$

Now, for any continuous bounded function $\psi$ and any $0 \leq s < 1$,
\[|P_{n + s}\psi(x) - \pi \psi| = |P_n (P_s\psi)(x) - \pi (P_s\psi)| \leq
\|P_n(x, .) - \pi\| ||P_s \psi||_{\infty}\] where  $\pi \psi$ stands for $\int \psi d\pi.$ Hence,
\[\lim_{n \to \infty} \|P_{n+s}(x, .) - \pi\|= 0 \]
 so assertion $(i)$ of the theorem holds.
The second assertion follows from the uniqueness of $\pi$ and Theorem \ref{thm:persistence2}.
\subsection{Proof of Proposition \ref{thm:robpersistence2}}
Suppose (\ref{eq:sdeperturb}) is a $\delta$-perturbation of (\ref{eq:sde})  with $X_0  = \tilde{X}_0 = x.$ Then for all $t \geq 0$
\begin{eqnarray*}
X_t - \tilde{X}_t &= &\int_0^t (X_s\circ F(X_s) - \tilde{X}_s\circ F(\tilde{X}_s)) ds + \int_0^t  (\tilde{X}_s\circ F(\tilde{X}_s) -\tilde{X}_s\circ  \tilde{F}(\tilde{X}_s)) ds +\\
&&\int_0^t (X_s\circ \Sigma(X_s) - \tilde{X}_s\circ \Sigma(\tilde{X}_s)) dB_s + \int_0^t  (\tilde{X}_s\circ \Sigma(\tilde{X}_s) - \tilde{X}_s\circ \tilde{\Sigma}(\tilde{X}_s)) dB_s\end{eqnarray*}
Let
\[v(t) = \E\left[\|X_t - \tilde{X}_t\|^2\right]. \] Then, by the Cauchy-Schwartz inequality, the fact that $\|X_s\|\le 1$, and the Ito isometry
 \begin{eqnarray*}
 v(t) &\leq &
  4 \E \left[ \| \int_0^t X_s\circ F(X_s) -\tilde{X}_s\circ  F(\tilde{X}_s)ds \|^2 + \|\int_0^t \tilde{X}_s\circ  F(\tilde{X}_s) - \tilde{X}_s\circ \tilde{F}(\tilde{X}_s) ds\|^2 \right. \\
 && \left.+  \|\int_0^t X_s\circ \Sigma(X_s) - \tilde{X}_s\circ \Sigma(\tilde{X}_s) dB_s \|^2+ \|\int_0^t  \tilde{X}_s\circ \Sigma(\tilde{X}_s) -\tilde{X}_s\circ  \tilde{\Sigma}(\tilde{X}_s) dB_s\|^2 \right]\\
 & \leq & 4  t \int_0^t \E [ \|X_s\circ F(X_s) - \tilde{X}_s\circ  F(\tilde{X}_s))\|^2 ]ds  +  4 t \int_0^t  \E [\|F(\tilde{X}_s) - \tilde{F}(\tilde{X}_s)\|^2] ds  \\
 &&+  4 \int_0^t \E[ \|X_s\circ\Sigma(X_s) - \tilde{X}_s\circ  \Sigma(\tilde{X}_s)\|^2 ]ds + 4 \int_0^t  \E[ \|\Sigma(\tilde{X}_s) - \tilde{\Sigma}(\tilde{X}_s)\|^2] ds \end{eqnarray*}
 Using the assumption and the Lipschitz continuity of $X\circ F(X)$ and $X\circ\Sigma(X)$ it follows that, for some constant $L,$
 \[v(t) \leq 4 t L \int_0^t v(s) ds + 4 t^2 \delta^2 + 4 L \int_0^t v(s) ds + 4 t \delta^2.\] Thus, for all $t \leq T$
 \[v(t) \leq A \int_0^t v(s) ds + B \delta^2\] where $A = 4 L (T+1)$ and $B = 4 T(T+1).$
 Hence, by Gronwall's lemma \[v(t) \leq e^{tA} B \delta^2\] for all $t \leq T.$
 The remainder of proof is  similar to the proof of \ref{thm:robpersistence}. The details are left to the reader.
\end{document}